\documentclass[12pt]{amsart}

\usepackage{amsmath, amsthm, amssymb,amscd}
\usepackage{fullpage, paralist,enumerate}
\usepackage{verbatim}
\usepackage[all]{xy}
\usepackage[parfill]{parskip}
\usepackage{graphicx}

\usepackage{afterpage}

\usepackage{amsmath}
\usepackage{amsfonts}
\usepackage{mathcomp}
\usepackage{textcomp}
\usepackage{amssymb}
\usepackage{latexsym}
\usepackage{graphicx}
\usepackage{courier}
\usepackage[english]{babel}
\usepackage{mathbbol}
\usepackage{multirow}
\usepackage{latexsym}
\usepackage{epsfig}
\usepackage{subfigure}
\usepackage{dcolumn}
\usepackage{bm}
\usepackage{colordvi}
\usepackage{color}
\usepackage{booktabs}
\usepackage{stmaryrd}
\usepackage{cite}
\usepackage[linesnumbered,commentsnumbered,ruled]{algorithm2e}
\usepackage[noend]{algpseudocode}
\usepackage{epstopdf}
\input xy
\xyoption{all}
\newcommand{\commentout}[1]{}

\newtheorem{thm}{Theorem}[section]

\newtheorem{prop}[thm]{Proposition}
\newtheorem{cor}[thm]{Corollary}
\newtheorem{ex}[thm]{Example}
\newtheorem{rmk}[thm]{Remark}

\newcommand{\nwc}{\newcommand*}

\nwc{\ben}{\begin{equation*}}
\nwc{\bea}{\begin{eqnarray}}
\nwc{\beq}{\begin{eqnarray}}
\nwc{\bean}{\begin{eqnarray*}}
\nwc{\beqn}{\begin{eqnarray*}}
\nwc{\beqast}{\begin{eqnarray*}}

\nwc{\eal}{\end{align}}
\nwc{\een}{\end{equation*}}
\nwc{\eea}{\end{eqnarray}}
\nwc{\eeq}{\end{eqnarray}}
\nwc{\eean}{\end{eqnarray*}}
\nwc{\eeqn}{\end{eqnarray*}}

\CompileMatrices

\theoremstyle{remark}

\nwc{\nn}{\nonumber}
\nwc{\mb}{\mathbf}
\nwc{\ml}{\mathcal}

\newcommand{\lt}{\left}
\newcommand{\rt}{\right}

\nwc{\vep}{\varepsilon}
\nwc{\ep}{\epsilon}
\nwc{\vrho}{\varrho}
\nwc{\orho}{\bar\varrho}
\nwc{\vpsi}{\varpsi}
\nwc{\lamb}{\lambda}
\nwc{\om}{\omega}
\nwc{\Om}{\Omega}
\nwc{\al}{\alpha}
\nwc{\sgn}{\mbox{\rm sgn}}

\nwc{\IA}{\mathbb{A}} 
\nwc{\bi}{\mathbf{i}}
\nwc{\ba}{\mathbf{a}}
\nwc{\bmb}{\mathbf{b}}
\nwc{\bo}{\mathbf{o}}
\nwc{\IB}{\mathbb{B}}
\nwc{\IC}{\mathbb{C}} 
\nwc{\ID}{\mathbb{D}} 
\nwc{\IM}{\mathbb{M}} 
\nwc{\IP}{\mathbb{P}} 
\nwc{\II}{\mathbb{I}} 
\nwc{\IE}{\mathbb{E}} 
\nwc{\IF}{\mathbb{F}} 
\nwc{\IG}{\mathbb{G}} 
\nwc{\IN}{\mathbb{N}} 
\nwc{\IQ}{\mathbb{Q}} 
\nwc{\IR}{\mathbb{R}} 
\nwc{\IT}{\mathbb{T}} 
\nwc{\IZ}{\mathbb{Z}} 

\nwc{\cE}{{\ml E}}
\nwc{\cP}{{\ml P}}
\nwc{\cQ}{{\ml Q}}
\nwc{\cL}{{\ml L}}
\nwc{\cX}{{\ml X}}
\nwc{\cW}{{\ml W}}
\nwc{\cZ}{{\ml Z}}
\nwc{\cR}{{\ml R}}
\nwc{\cV}{{\ml V}}
\nwc{\cT}{{\ml T}}
\nwc{\crV}{{\ml L}_{(\delta,\rho)}}
\nwc{\cC}{{\ml C}}
\nwc{\cO}{{\ml O}}
\nwc{\cA}{{\ml A}}
\nwc{\cK}{{\ml K}}
\nwc{\cB}{{\ml B}}
\nwc{\cD}{{\ml D}}
\nwc{\cF}{{\ml F}}
\nwc{\cS}{{\ml S}}
\nwc{\cM}{{\ml M}}
\nwc{\cG}{{\ml G}}
\nwc{\cH}{{\ml H}}
\nwc{\bk}{{\mb k}}
\nwc{\bn}{{\mb n}}
\nwc{\bp}{{\mb p}}
\nwc{\bq}{{\mb q}}
\nwc{\bz}{\mb z}
\nwc{\bl}{{\mb l}}
\nwc{\bj}{{\mb j}}
\nwc{\bs}{{\mb s}}
\nwc{\by}{\mathbf{h}}
\nwc{\bZ}{\mathbf{Z}}
\nwc{\bF}{\mathbf{F}}
\nwc{\bE}{\mathbf{E}}
\nwc{\bV}{\mathbf{V}}
\nwc{\bY}{\mathbf Y}
\nwc{\br}{\mb r}
\nwc{\pft}{\cF^{-1}_2}
\nwc{\bU}{{\mb U}}
\nwc{\bG}{{\mb G}}
\nwc{\bg}{\mathbf{g}}
\nwc{\mbf}{\mathbf{f}}
\nwc{\mbe}{\mathbf{e}}
\nwc{\be}{\mathbf{e}}
\nwc{\ind}{\operatorname{I}}
\nwc{\mbx}{\mathbf{f}}
\nwc{\bb}{\mathbf{g}}
\nwc{\xmax}{f_{\rm max}}
\nwc{\xmin}{f_{\rm min}}
\nwc{\suppx}{\hbox{\rm supp} (\mbf)}
\nwc{\cI}{\IZ^2_N}
\nwc{\chis}{{\chi^{\rm s}}}
\nwc{\chii}{{\chi^{\rm i}}}
\nwc{\pdfi}{{f^{\rm i}}}
\nwc{\pdfs}{{f^{\rm s}}}
\nwc{\pdfii}{{f_1^{\rm i}}}
\nwc{\pdfsi}{{f_1^{\rm s}}}
\nwc{\thetatil}{{\tilde\theta}}
\nwc{\red}{\color{red}}
\nwc{\blue}{\color{blue}}

\nwc{\prox}{\hbox{prox}}
\nwc{\diag}{\hbox{\rm diag}}
\nwc{\supp}{{\hbox{\rm supp}}}

\nwc{\sloc}{J_{\rm f}}
\nwc{\bu}{{\mb u}}
\nwc{\bv}{{\mb v}}
\nwc{\cU}{\mathcal{U}}
\nwc{\cN}{\mathcal{N}}
\nwc{\bN}{\mathbf{N}}
\nwc{\mbm}{\mathbf{m}}
\nwc{\bw}{\mathbf{w}}
\nwc{\bom}{\mathbf{w}}
\nwc{\bt}{\mathbf{t}}
\nwc{\z}{y}
\nwc{\cY}{\mathcal{Y}}
\nwc{\bM}{\mathbf{M}}
\nwc{\half}{{1\over 2}}
\nwc{\Sf}{S_{\rm f}}
\nwc{\Jf}{J_{\rm f}}
\nwc{\nul}{\hbox{\rm null}_\IR}
\nwc{\spanR}{\hbox{\rm span}_\IR}
\nwc{\Arg}{\hbox{\rm Arg~}}
\nwc{\fdr}{S_{\rm f}}
\nwc{\phase}[1]{\exp\lt[i\measured #1\rt]}

\nwc{\im}{{\rm i}}

\nwc{\lb}{\llbracket}
\nwc{\rb}{\rrbracket}
\nwc{\modpi}{{{\rm mod}\,2\pi}}

\begin{document}

 \title{
Uniqueness Theorems for Tomographic Phase Retrieval with Few Coded Diffraction Patterns
}

\author{Albert Fannjiang 
 \address{
Department of Mathematics, University of California, Davis, California  95616, USA. Email:  {\tt fannjiang@math.ucdavis.edu}
} 
}


\maketitle 

\begin{abstract} 3D tomographic phase retrieval under the Born approximation  for  discrete objects supported on a $n\times n\times n$ grid is analyzed.  It is proved that $n$ projections are sufficient and necessary for unique determination by computed tomography (CT) with full projected field measurements and that $n+1$ coded projected diffraction patterns are sufficient for unique determination, up to a global phase factor, in tomographic phase retrieval. Hence $n+1$ is nearly, if not exactly, the minimum number of diffractions patterns needed for 3D tomographic phase retrieval under the Born approximation.

 \end{abstract}


\section{Introduction}\label{sec:intro}

Tomography is a commonly used  method in a wide range of applications such as computed tomography  \cite{CT}, 3D diffractive imaging  \cite{Devaney} and quantum state measurement \cite{Leonhardt}. 

Mathematically speaking, the forward model of tomography is based on various approximations of the nonlinear inverse scattering formulation (for example,  the Lippmann-Schwinger integral equation). A simplification common to all current tomographic methods (except for geophysical applications) is based on either the Born or the Rytov approximation. The latter reduces to computed tomography (CT) in the limit of geometrical optics.  The inversion methods of CT, which ignores the diffraction and scattering effects,  have been well studied and documented \cite{CT}. On the other hand, the phase-unwrapping problem inherent to
the Rytov approximation (see Section \ref{sec:Born})  is a largely unsolved problem and a major road block to its implementation  \cite{Devaney}. 

An additional  complication  occurs in X-ray, optical scattering \cite{fel}, electron diffraction \cite{ED, Frank}  as well as quantum state tomography \cite{Leonhardt}, where only intensity measurements can be performed. This  gives  rise to the phase problem which requires phase retrieval techniques for solutions \cite{acta}. 

This brief note considers the imaging set-up based on the Born approximation where diffraction patterns (hence intensity-only measurements) in various directions are measured and used to determine the 3D  object. 

In particular, we address the uniqueness question: Under what measurement schemes and with how many diffraction patterns, can one determine the 3D object uniquely (up to a global phase factor)? 

To answer this question in a quantitative way, it is instructive  (even imperative) to work with a discrete setting. After introducing the Born-projection approximation in Section \ref{sec:Born} and laying out the discrete framework in Section \ref{sec:discrete}, we recall some basic results about diffraction patterns in Section \ref{sec:patterns}, in particular how the use of a random mask can improve the quality of the measurement data (see also Remark \ref{rmk4.3}). In Section \ref{sec:unique} we first prove that with a random mask in the measurement of diffraction patterns, the tomographic phase retrieval problem reduces to that of CT
modulo a simple ambiguity  (Theorem \ref{tom}). We then eliminate this ambiguity by deploying a sufficiently diverse set of $n+1$ projections under  the prior constraint that the object does not become part of a line segment in any projection in the measurement scheme. As the uniqueness condition of $n$ projections required for the standard CT  
(Theorem \ref{thm:CT}) sets a lower bound on  the number of diffraction patterns for tomographic phase retrieval, the uniqueness condition of $n+1$ diffraction patterns (Theorem \ref{tom2}) is nearly optimal.  We conclude with several remarks in Section \ref{sec:conclude}.

{  \section{Born and projection approximations} \label{sec:Born}
\begin{figure}
\centering
\includegraphics[width=14cm]{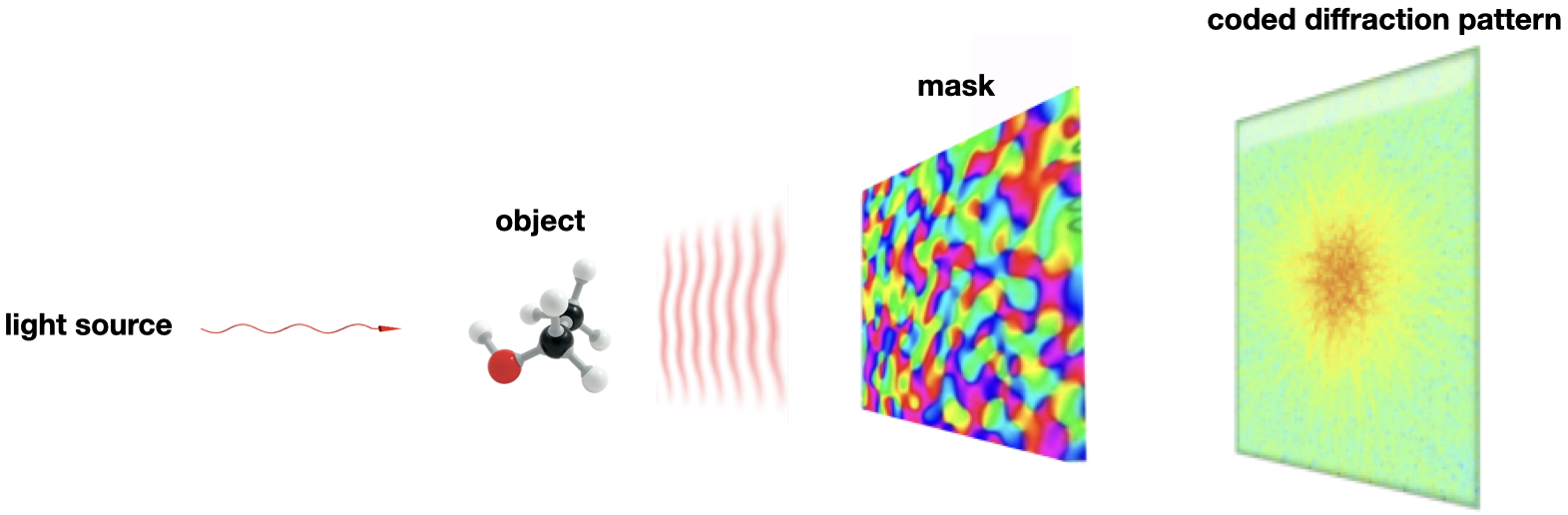}
\caption{Diffraction pattern coded by a random mask placed behind the object. Different projections can be implemented by orientating the object in the corresponding direction, with the measurement set-up fixed. See, e.g. \cite{2013} for an similar experimental set-up}
\commentout{
\includegraphics[width=8cm]{coded}
\caption{  Diffraction pattern coded by a random mask placed behind the object (courtesy of \cite{Candes}. Copyright ©2015 Society for Industrial and Applied Mathematics. Reprinted with permission. All rights
reserved). Different projections can be implemented by orientating the object in the corresponding direction, with the measurement set-up fixed. See, e.g. \cite{2013} for such an experimental set-up.}
}
\label{fig1}
\end{figure}

In scattering theory,  the full field $u=u_i+u_s$ is written as  the sum of the incident field $u_i$ and the scattered field $u_s$.
In the continuum setting, the full field $u(\br)$ is governed by the Lippmann-Schwinger equation
\beq
u(\br)=u_i(\br)+\int d\br' G(\br-\br') f(\br') u(\br')\label{LS}
\eeq
where $f$ is the inhomogeneity, also called {\em scattering potential},  and $G$ is the Green's function of the free-space Helmholtz equation \cite{Devaney}. 

\commentout{The right hand side of \eqref{LS} can be iterated to create the Born series with the first two terms given by 
\beqn
u(\br)=u_i(\br)+\int d\br' G(\br-\br') f(\br') u_i(\br')+\int d\br' G(\br-\br') f(\br')\int d\br'' G(\br'-\br'') f(\br'')+ \cdots \label{Born}
\eeqn
with higher order terms involving higher moments of $f$. }
Under the weak scatter assumption $|u_s|\ll |u_i|, $ 
$u$ in the (first-order) Born approximation is given by 
\beq
u(\br)=u_i(\br)+\int d\br' G(\br-\br') f(\br') u_i(\br').\label{Born1}
\eeq  
Under the Fresnel approximation  (with the $z$-axis as the optical axis, say), 
\beq\label{Fresnel}
G(\br)={-1\over 4\pi}{e^{\im \kappa|\br|}\over |\br|}\approx {-1\over 4\pi |z|}e^{\im \kappa|z|} e^{\im {\kappa\over 2} {x^2+y^2\over |z|}}
\eeq
and hence \eqref{Born1} becomes
\beq
\label{Fresnel1}
u_i(\br)-{e^{\im \kappa z}\over 4\pi} \int dx' dy' \int d z'{f(x',y',z')\over |z-z'|}e^{\im {\kappa\over 2} {(x-x')^2+(y-y')^2\over |z-z'|}} e^{-\im \kappa z'} u_i(x', y', z').
\eeq

We think of the scattering process as consisting of two stages: First, the plane wave ($u_i(\br)=e^{\im \kappa z}$)  illuminates 
and exits the scattering object; second, the exit wave transmits through a mask (located at $z=0$) and
propagates toward the detector. 

In the first stage,   consider the high Fresnel number regime
\beq\label{high}
N_F={\ell^2 \over \lambda z_0}\gg 1,
\eeq
where $\ell$ is the typical size to be resolved, $\lambda $ the wavelength and $z_0$ the thickness of the object.
In this limit \eqref{high}, 
\[
 {-\im \kappa\over 2\pi |z-z'|} e^{\im {\kappa\over 2} {(x-x')^2+(y-y')^2\over |z-z'|}}\longrightarrow \delta(x-x', y-y'),\quad \mbox{as}\,\, N_F\to\infty
 \]
 for all $z\neq z'$,  the exit wave
\eqref{Fresnel1} at $z=0$ is approximated by 
\beq
\label{Fresnel2}
v_B(x,y)= 1-{\im\over 2\kappa}  \int d z'{f(x,y,z')}.
\eeq
The right hand side of \eqref{Fresnel2} is the {\em projection approximation} under the first-order Born assumption. On the other hand, the exit wave with the Rytov approximation is given by 
\beqn
\label{Fresnel3}
v_R(x,y)= \exp\lt[-{\im\over 2\kappa}  \int d z'{f(x,y,z')}\rt].
\eeqn
Since
\beq
v_B-1=\ln v_R\quad \hbox{mod}\,\, 2\pi\im,
\eeq
the Born scattered field is the unwrapped phase of the Rytov approximation. 

In short, 
the Born-projection approximation is the linear approximation of the Rytov-projection approximation and both approximations employ the projection approximation \cite{BR}.

The projection approximation corresponds to light propagation through the scatterer in parallel straight lines.
Its validity, however,  depends on the spatial resolution of the imaging system as follows.  \commentout{``A coarse enough imaging resolution would prevent the detection of fine diffraction effects, hence the projection approximation would hold even in a case where, for the same sample under the same illumination conditions, the projection approximation would be invalid when the scattered intensity is measured with a position-sensitive detector having a spatial resolution that is less coarse" \cite{phase-contrast}.}

Radiation of wavelength $\lamb$  scattered by  features of  size $\ell$, that are to be resolved,  would have a maximum diffraction angle of the order of
$\Delta\theta =\lambda/\ell.$ Hence the maximum spread of the radiation at the exit plane would be $\Delta\theta z_0$
where $z_0$ is the thickness of the sample. The projection approximation is valid if the spread is much smaller than the resolution, i.e. 
\beq
\lambda z_0/\ell  \ll \ell.\label{proj}
\eeq
which is exactly equivalent to the high Fresnel number regime \eqref{high}
 \cite{phase-contrast}. 


At the second stage, the exit wave $v_B$  is first multiplied by the mask function  $\mu$  and then propagates into the far-field as $\cF (\mu \cdot  v_B)$ where $\cF$ is the Fourier transform in the transverse variables. The measured coded diffraction pattern  $|\cF (\mu \cdot v_B)|^2$ is given by
\beq
|\cF (\mu \cdot v_B)|^2&=& |\cF(\mu)|^2+{1\over\kappa}\Im\{{ \overline{\cF\mu}\cdot \cF(\mu \int f dz')}\}+{1\over 4\kappa^2} |\cF(\mu\int fdz')|^2\label{fpt}
\eeq
where $\Im$ denotes the imaginary part. 
This technique, with or without coded aperture,  is sometimes called the  {\em propagation-based phase contrast} method \cite{phase-contrast}. 

Tomographic microscopy based on the  linear forward model ignoring the nonlinear term $ |\cF(\mu\int dz' f)|^2 $ on the right hand side of \eqref{fpt} is {a form of bright-field imaging}  (see  \cite{2018,microscopy}). {As
\eqref{fpt} represents the interference pattern between the reference wave $\cF(\mu)$ and the masked object wave $-\im\cF( \mu\int f dz')/(2\kappa)$, reconstruction from the linear term in \eqref{fpt} can be performed by conventional holographic techniques \cite{Wolf69,Wolf70}. }

{Adopting the dark-field mode of  imaging}
(see \cite{FPT} where the pupil or probe function plays the role of coded aperture), we focus on the more challenging nonlinear term
 as the measurement data and analyze the inherent information content therein. 
The (nonlinear) combination of the coded aperture and the Fourier transform is a key ingredient of our approach. 
The next key ingredient to an information-based approach is discretization. 

}

\section{Discrete tomography}\label{sec:discrete}

To motivate the discrete setup, consider the continuum setting. It is a classical result that  a compactly supported function on, e.g. the cube, is uniquely determined by  the Fourier transform (magnitude \& phase) in any infinite set of projections (\cite{Helgason}, Proposition 7.8) while
for any finite set of projections, counterexamples to unique determination can be constructed (\cite{Helgason}, Proposition 7.9).

As a consequence, uniqueness with Fourier {\em intensity} data in the continuum setting  would  require additional assumptions besides an infinite number of projections. It is not currently known, however, what additional assumptions are  needed to  guarantee uniqueness with intensity-only measurements. 

Working with a discrete set-up we aim  to derive a quantitative, information-based theory of uniqueness. 
To this end,  we  adopt the framework of \cite{discrete-X} whose main advantage is preserving the fundamental  Fourier slice theorem (Theorem \ref{thm:slice}).

For simplicity, we choose the physical units so that   $\kappa =2\pi.$
Let $\lb k,l\rb$ denote the integers
between and including the integers $k$ and $l$. 
We define a 3D $n\times n \times n$ object as the set 
\beq
\label{1.1}
 f=\{f(i,j,k)\in \IC: i,j,k\in \IZ_n\}
 \eeq
 where
\beq
\label{1.2}
\IZ_n&=&\left\{\begin{array}{lll}
\lb-n/2, n/2-1\rb && \mbox{if $n$ is an even integer;}\\
 \lb-(n-1)/2, (n-1)/2\rb && \mbox{if $n$ is an odd integer.}
 \end{array}\rt.
\eeq

We define three families of line segments, the $x$-lines, $y$-lines, and $z$-lines. Formally, a $x$-line, denoted by $\ell_{x(\alpha,\beta)}(c_1,c_2)$,  is defined as 
\beq
\label{1.3'}
\ell_{x(\alpha,\beta)}(c_1,c_2): \lt[\begin{matrix}
y\\
z\end{matrix}\rt]=\lt[\begin{matrix}\alpha x+c_1\\  \beta x+c_2
\end{matrix}\rt] && {  c_1, c_2\in \IZ_{2n-1}}, \quad x\in \IZ_n
\eeq
 To avoid wraparound of $x$-lines with $|\alpha|,|\beta|\le 1$, we can zero-pad $f$ in a larger lattice $\IZ^3_p$ with $p\ge 2n-1.$ This is particularly important when it comes to define the X-ray transform by a line sum (cf. \eqref{2.8}-\eqref{2.10}) without wrapping around the object domain. 

\commentout{
More generally, we can extend the set-up in \cite{discrete-X} by letting 
\[
 |\alpha|,|\beta|\le \gamma, 
 \]
 for a constant $\gamma\ge 1$ with
 \beq
 \label{3.1}
 p\ge (\gamma+1)(n-1)+1.
 \eeq

  To fix the idea, we assume $p=2n-1$ and $|\alpha|, |\beta|\le 1$ in the present paper.
  }

Similarly, a $y$-line and a $z$-line are defined as
\beq
\ell_{y(\alpha,\beta)}(c_1,c_2):  \lt[\begin{matrix}
x\\
z\end{matrix}\rt]=\lt[\begin{matrix}\alpha y+c_1\\  \beta y+c_2
\end{matrix}\rt]
 && c_1, c_2\in \IZ_{2n-1}, \quad y\in \IZ_n,\\
\ell_{z(\alpha,\beta)}(c_1,c_2): 
 \lt[\begin{matrix}
x\\
y\end{matrix}\rt]=\lt[\begin{matrix}\alpha z+c_1\\  \beta z+c_2
\end{matrix}\rt] && c_1, c_2\in \IZ_{2n-1}, \quad z\in \IZ_n. 
\eeq
\commentout{
An $x$-line with $\alpha\neq 0$ can also be expressed as a $y$-line as
\beqn
x=y/\alpha-c_1/\alpha,&& z=\beta y/\alpha-\beta c_1/\alpha+c_2
\eeqn
and an $x$-line with $\beta\neq 0$ can also be expressed as a $z$-line as
\beqn
x=z/\beta-c_2/\beta,&& y=\alpha z/\beta-\alpha c_2/\beta+c_1
\eeqn
}
We denote the sets of all $x$-lines, $y$-lines, and $z$-lines by $\cL_x, \cL_y,$ and $\cL_z$, respectively. 

\commentout{
The intersection of $\cL_x$ and $\cL_y$ consists of those $x$lines or $y$-lines with $\alpha=\pm 1.$ For $|\alpha|=1$, \eqref{1.3'} can be written as
\beq
\label{1.6}
 \lt[\begin{matrix}
x\\
z\end{matrix}\rt]=\lt[\begin{matrix} y/\alpha-c_1/\alpha\\  \beta y/\alpha-\beta c_1/\alpha+c_2
\end{matrix}\rt],\quad \alpha=\pm 1.
\eeq
Likewise, the intersection of $\cL_x$ and $\cL_z$ consists of those $x$-lines or $z$-lines with $\beta=\pm 1.$ For $|\beta|=1$, \eqref{1.3'} can be written as
\beq\label{1.6'}
 \lt[\begin{matrix}
x\\
y\end{matrix}\rt]=\lt[\begin{matrix}z/\beta -c_2/\beta\\  \alpha z/\beta +c_1-\alpha c_2/\beta 
\end{matrix}\rt],\quad\beta=\pm 1. 
\eeq
Intersection of $\cL_y$ and $\cL_z$ has an analogous form.
}

Also, we denote the family of lines that corresponds to a fixed pair $(\alpha,\beta)$  and variable intercepts $(c_1,c_2)$ by $\ell_{x(\alpha,\beta)}, \ell_{y(\alpha,\beta)}$ and $\ell_{z(\alpha,\beta)}$ for a family of parallel $x$-lines, $y$-lines, and $z$-lines, respectively.
Note that $\ell_{x(1,\beta)}=\ell_{y(1,\beta)}, \ell_{x(\alpha,1)}= \ell_{z(1,\alpha)}$ and
$\ell_{y(\alpha,1)}= \ell_{z(\alpha,1)}.$

Let $f_x$ be the continuous interpolation of $f$ in the directions perpendicular to $x$ as follows:
\beq\label{1.5}
f_x(i,y,z)&=& \sum_{j\in \IZ_n}\sum_{k\in \IZ_n} f(i,j,k) D_p(y-j)D_p(z-k),\quad y,z\in \IR
\eeq
where $D_p$ is the $p$-periodic Dirichlet kernel given by
\beq
D_p(t)={1\over p} \sum_{l\in \IZ_{p}} e^{\im 2\pi l t/p}
&=&\lt\{\begin{matrix}1,& t=mp,\quad m\in \IZ\\
{\sin{(\pi t)}\over p\sin{(\pi t/p)}},& \mbox{else}.
\end{matrix}\rt.\nn
\eeq
In particular, $D_p(t)=0$  for $t\in \IZ/\IZ_p$, i.e. $[D_p(i-j)]_{i,j\in \IZ_p}$ is the $p\times p$ identity matrix. 

Similarly we define the interpolation of $f$ perpendicular to $y$ and $z$, respectively, as
\beq
f_y(x,j,z)&=& \sum_{i\in \IZ_n}\sum_{k\in \IZ_n} f(i,j,k) D_p(x-i)D_p(z-k),\quad x,z\in \IR;\label{1.5'}\\
f_z(x,y,k)&=& \sum_{i\in \IZ_n}\sum_{j\in \IZ_n} f(i,j,k) D_p(x-i)D_p(y-j), \quad x,y\in \IR.\label{1.5''}
\eeq
By interpolating from the grid points \eqref{1.5}-\eqref{1.5''}, we have extended $f$ from $\IZ_p^3$ to the hyperplanes
$x=i$, $y=j$ or $z=k,$ where $i,j,k\in \IZ_p$.

\commentout{
Deploying  the discrete set of slopes 
\beq\label{1.15}
\alpha,\beta\in \{2i/n: i\in \IZ_n\}
\eeq
enforces  intersection of each line with some grid points of the discrete object $f$. 
However, with the interpolation \eqref{1.5}, we can consider more general set $S$ of slopes with $n$ distinct values in the range $(-1,1)$. 
}

\commentout{The Dirichlet kernel performs trigonometric interpolation of length $p$ by zero padding the discrete object $f$ to length $p$ in the corresponding variable. We choose $ p\ge 2n-1$ as the length of the interpolation kernel to avoid wraparound, due to the periodic nature of the Dirichlet kernel, and  view the object as defined in $\IZ^3_{p}$ with zero padding outside of $ \IZ_n^3$. 
}

The main, and only, purpose for interpolating the discrete object is to make possible the definition of a diversified set of the discrete X-ray transforms. 
Having  extended the domain of $f$ to the hyperplanes
$x=i$, $y=j$ or $z=k,$ where $i,j,k\in \IZ_{2n-1}$, 
we define the discrete X-ray transforms as the line sums
\beq
 \label{2.8} 
 f_{x(\alpha,\beta)}(c_1,c_2)&=& \sum _{i\in \IZ_n} f_x(i,\alpha i+c_1,\beta i+c_2),\\
f_{y(\alpha,\beta)}(c_1,c_2)&=& \sum _{j\in\IZ_n} f_y(\alpha j+c_1,j, \beta j+c_2)\label{2.9}\\
f_{z(\alpha,\beta)}(c_1,c_2)&=& \sum _{k\in \IZ_n} f_z(\alpha k+c_1,\beta k+c_2, k)\label{2.10}
\eeq
with $ c_1, c_2\in \IZ_{2n-1}$. With zero-padding, we take $\IZ_{p}^2,$ $ p\ge 2n-1,$ as the domain of the X-ray transforms. 

Without the interpolation \eqref{1.5}-\eqref{1.5''}, the discrete X-ray transforms are not well-defined except for  $(\alpha,\beta)=(\pm 1,0), (0,\pm 1), (\pm 1,\pm 1).$ 
For simplicity of terminology, we shall refer to X-ray transforms simply as {\em projections}.

\commentout{with $|\alpha|,|\beta|\le 1$, where the range of  $c_1, c_2$ is $\IZ_{n}$. 
With $|\alpha|, |\beta|<1$, the arguments $\alpha i+c_1,\beta i+c_2$ on the right hand side of \eqref{2.8} remain in the range
$[-(p-1)/2, (p-1)/2]$ for $p\ge 2n-1$. }

\commentout{Hence, for a given object $f$ and a line $\ell$, we define the discrete X-ray transform of $f$ for the line $\ell$ as
\beq
Pf (\ell)&=& \lt\{\begin{matrix} 
P_x f(\ell),&& \ell\in \cL_x\\
P_y f(\ell),&& \ell\in \cL_y\\
P_z f(\ell),&& \ell\in \cL_z. 
\end{matrix}
\rt.
\eeq
}

The 3D Fourier transform $\widehat f$ of the object $f$, supported in $\IZ_n^3\subset \IZ_p^3$,  is given by
\beq
\label{1.10'}
\widehat f(\xi,\eta,\zeta)&=&\sum_{i,j,k\in \IZ_n} f(i,j,k)e^{-\im 2\pi(\xi i+\eta j+ \zeta k)/p}.
\eeq
Note that $\widehat f$  in \eqref{1.10'} is a $p$-periodic function band-limited to $\IZ_n^3$. The associated 1-D and 2-D (partial) Fourier transforms are similarly defined $p$-periodic band-limited functions. 

The Fourier slice theorem concerns the 2-D discrete Fourier transform $\widehat f_x(\alpha,\beta)$, defined as
\beq
\widehat f_{x(\alpha,\beta)}(\eta,\zeta)&=& \sum_{j,k\in \IZ_{n}} f_{x(\alpha,\beta)}(j,k)e^{-\im 2\pi(\eta j+ \zeta k)/p}
\eeq
and the 3-D discrete Fourier transform given in \eqref{1.10'}. 

The following Fourier slice theorem resembles that of the continuous case \cite{CT} and plays a central role in the framework of discrete tomography.
\begin{thm}\label{thm:slice}\cite{discrete-X}
(Fourier slice theorem) For a given family of $x$-lines $\ell_x(\alpha,\beta)$ with fixed slopes $(\alpha,\beta)$ and variable intercepts $(c_1,c_2)$. Then the 2D discrete Fourier transform  $\widehat f_{x(\alpha,\beta)}$ of the $x$-projection $f_{x(\alpha,\beta)}$ and the 3D discrete Fourier transform $\widehat f$ of the object $f$ satisfy the equation
\beq\label{proj-x}
\widehat f_{x(\alpha,\beta)}(\eta,\zeta)&=& \widehat f(-\alpha \eta-\beta \zeta, \eta,\zeta).
\eeq
Likewise, we have
\beq
\widehat f_{y(\alpha,\beta)}(\xi,\zeta)&=& \widehat f(\xi, -\alpha \xi-\beta \zeta, \zeta),\\ 
\widehat f_{z(\alpha,\beta)}(\xi,\eta)&=& \widehat f(\xi,\eta,  -\alpha \xi-\beta \eta).\label{proj-z}
\eeq
\end{thm} 

\subsection{Continuum limit}
One can  justify the above discrete framework, especially the interpolation scheme \eqref{1.5}-\eqref{1.5''} and the related  line average 
\eqref{2.8}-\eqref{2.10}, from the perspective of continuum limit. 

 Suppose the discrete object $f$ above is the restriction of some smooth function $f_*$ supported on $[-1/2,1/2]^3$ in the sense that 
\[
f(i,j,k)=f_*\Big({i\over n}, {j\over n}, {k\over n}\Big)
\]
or some local average of $f_*$ about each grid point.
\commentout{With this rescaling the sampling pattern changes from \eqref{samples} to 
\beq\label{samples2}
\Big\{(w_1,w_2): w_1, w_2 = -n+{1}, -n+{1\over 2},\ldots,  -1, -\half, 0, \half,  \ldots,  n-{1} \Big\}
\eeq
}
As $p\to \infty$, the Dirichlet kernel has the limit
\[
\lim_{p\to\infty} pD_p(pt)=\delta(t), 
\]
Dirac's delta function. For a sufficiently smooth $f_*$, 
 the right hand side of \eqref{1.5}, after proper normalization, approaches the limit
\beqn
 \int  f_*(x,y',z') \delta (y-y')\delta(z-z') dy' dz'&=&f_*(x,y,z) 
\eeqn
 In other words, the interpolation becomes exactly an identity  in the continuum limit. Likewise, the discrete X-ray transforms \eqref{2.8}-\eqref{2.10}, after proper normalization, become  line integrals (i.e. the continuous X-ray transforms). 

Finally, in the continuum limit, Theorem \ref{thm:slice} gives rise to the standard Fourier slice theorem. 
In other words, the discrete framework is a structure-preserving discretization of the continuous setting. 

\section{Diffraction patterns}\label{sec:patterns}

For ease of notation, we denote  by $\bt$ the direction of projection, $x(\alpha,\beta), y(\alpha,\beta)$ or $z(\alpha,\beta)$. Let $\cT$ denote the set of directions $\bt$ employed in the tomographic measurement. {  Let $p=2n-1$}. 

Let the Fourier transform of the projection $f_\bt$ be written as 
\[
F_\bt(e^{-\im 2\pi\bw})=\sum_{\bn\in \IZ_p^2} e^{-\im 2\pi \bn\cdot\bw} f_\bt(\bn),\quad \bw\in \Big[-\half,\half \Big]^2,
\]
where  $f_\bt$ vanishes outside 
 $\IZ^2_{n}$. {In the absence of a random mask ($\mu\equiv 1$)}, the continuous diffraction pattern in the far field  can be written as  \beq
|F_\bt(e^{-\im 2\pi\bw})|^2= \sum_{\bn\in   \IZ_{2p-1}^2}\lt\{\sum_{\bn'\in \IZ_p^2} f_\bt(\bn'+\bn)\overline{f_\bt(\bn')}\rt\}
   e^{-\im 2\pi \bn\cdot \bom}, \quad \bom\in  \Big[-\half,\half \Big]^2,
   \label{auto}
   \eeq
\cite{unique}. 
 Here and below the over-line notation means
complex conjugacy. The expression in the brackets  in \eqref{auto} is the autocorrelation function of $f_\bt$.

The diffraction patterns are then uniquely determined by  sampling on the grid
\beq
\label{nyquist1}
\bw\in {1\over 2p-1} \IZ_{2p-1}^2
\eeq
or by Kadec's $1/4$-theorem  on any following  irregular grid \cite{Young}
\beq
\label{nyquist2}
\{\bw_{jk},\,\, j,k\in \IZ_{2p-1}: |(2p-1)\bw_{jk}-(j,k)|<1/4\}.  
\eeq
With the Nyquist, regular \eqref{nyquist1} or irregular \eqref{nyquist2},  sampling, the diffraction pattern contains the same information as does  the autocorrelation function of $f_\bt$. 

\subsection{Inherent ambiguities}

The following standard result explicates all the ambiguities corresponding to the same diffraction pattern.

\begin{prop}
\label{prop1}\cite{Hayes}
Let the $z$-transform 
$F_\bt(\bz) = \sum_{\bn\in \IZ^2_{p}} f_\bt(\bn) \bz^{-\bn}$ 
 be
given by
\beq
F_\bt(\bz)=\alpha \bz^{-\mbm} \prod_{k=1}^q F_k(\bz),\quad  \mbm\in \IN^2, \quad \alpha\in \IC\label{21}
\eeq
where $F_k, k=1,\dots,q,$ are 
non-monomial irreducible polynomials. Let $G_\bt(\bz)$ be
the $\bz$-transform of another finite array $g_\bt(\bn)$. 
Suppose 
$|F_\bt(e^{-\im 2\pi\bw})|=|G_\bt (e^{-\im 2\pi\bw})|,\forall \bw\in [0,1]^2$. Then \beq
\label{21'}
G_\bt(\bz)=|\alpha| e^{\im \theta} \bz^{-\bq}
\lt(\prod_{k\in I} F_k(\bz)\rt)
\lt(\prod_{k\in I^c} \overline{F_k(1/\bar{\bz})}\rt),\quad\mbox{for some}\quad \bq\in \IN^2,\, \theta\in \IR, 
\eeq
 where $I$ is a subset of $\{1,2,\dots,q\}$. 
\end{prop}

\begin{rmk}\label{rmk1}
The undetermined monomial factor $\bz^{-\bq}$ in \eqref{21'} corresponds to the 
translation invariance of the Fourier intensity data while 
the altered factors $\overline{F_k(1/\bar{\bz})}$ corresponds to the conjugate inversion invariance
of the Fourier intensity data (see Corollary \ref{cor1} below). 
The conjugate inversion of $f_\bt$, called the twin image, is defined by  $\mb{\rm Twin} (f_\bt) (\bn)= \bar f_\bt(-\bn)$. 

\end{rmk}

{  Next consider  a random mask
$
\mu(\bn)=e^{\im \phi(\bn)}
$
where  $\phi(\bn)$ are independent, continuous random variables over $[-\pi,\pi)$.
To fix the idea, let  the mask be placed between the object and the detectors (Figure \ref{fig1}) so that the measured diffraction pattern is the intensities of the Fourier transform of the masked projection 
 $\widetilde f_\bt(\bn) =f_\bt(\bn) \mu(\bn)$, i.e. the $\mu$-coded diffraction pattern. 
 }

 Let
 $f_\bt$ be not part of a line object. An object is part of a line object  if its support  is a subset of  a line. 
 Consequently, the masked projection $\widetilde f_\bt(\bn)$ is not part of a line object. 

Recall  \cite{unique} that the $z-$transform of the non-line masked object projection
is irreducible, up to a monomial as stated below. 
\begin{prop}\cite{unique}
Suppose $f_\bt$ is not a line object and let $\mu$ be the phase mask with phase at each point
 continuously and independently distributed over $[-\pi,\pi)$. 
Then with probability one 
the $z$-transform of the masked object $\widetilde f_\bt =f_\bt\odot \mu$ does not have any non-monomial irreducible
polynomial factor. 
\label{prop2}
\end{prop}

The masked object is also called the {\em exit wave} in the parlance of optics literature. 
In other words, a coded diffraction pattern is just the plain diffraction pattern of
a masked object. 


The following corollary will be useful for subsequent analysis. 
\begin{cor}\label{cor1} Under the assumptions of Proposition \ref{prop2}, if another masked object projection $\widetilde g_\bt:=\nu g_\bt$  produces
the same diffraction pattern as $\widetilde f_\bt=\mu \odot f_\bt$, then for some $\bp$ and $\theta$
\beq\label{r1}
 \widetilde f_\bt(\bn+\bp)&=& e^{-\im \theta}\widetilde g_\bt (\bn) \quad\mbox{or}\quad
 e^{\im \theta}\, \mb{\rm Twin}(\widetilde g_\bt)(\bn)
\eeq
for all $ \bn$.
\end{cor}
\begin{proof} Let $\widetilde F_\bt$ and $\widetilde G_\bt $ be the $z$-transforms of $\widetilde f_\bt$ and $\widetilde g_\bt$, respectively. 
By Proposition  \ref{prop2} and \eqref{21'}, 
\beqn
\widetilde G_\bt(\bz)=e^{\im \theta} \bz^{-\bp}
\widetilde F_\bt (\bz)\quad\mbox{or}\quad e^{\im \theta} \bz^{-\bp}
\overline{\widetilde F_\bt (1/\bar{\bz})},\quad\mbox{for some}\,\,\bp,\, \theta\,\,\mbox{and all}\,\,\bz. 
\eeqn
which after substituting $\bz=\exp{(-\im 2\pi \bw)}$ becomes 
\beqn
\widetilde G_\bt (e^{-\im 2\pi\bw})=e^{\im \theta} e^{\im \bw\cdot\bp}
\widetilde F_\bt (e^{-\im 2\pi\bw})\quad\mbox{or}\quad e^{\im \theta} e^{\im \bw\cdot\bp}
\overline{\widetilde F_\bt (e^{-\im 2\pi\bw})},\quad\mbox{for some}\,\,\bp,\, \theta\,\,\mbox{and all}\,\,\bz. 
\eeqn

Note that $\widetilde G_\bt(e^{-\im 2\pi\bw})$ and $\widetilde F_\bt(e^{-\im 2\pi\bw})$ are the Fourier transforms of $\widetilde g_\bt$ and
$\widetilde f_\bt$, respectively. 
Therefore in view of Remark \ref{rmk1} we have 
\beqn
\widetilde g_\bt(\bn)&=& e^{\im \theta} \widetilde f_\bt(\bn - \bp)\quad\mbox{\rm or}\quad  e^{\im \theta}\, \mbox{\rm Twin}(\widetilde f_\bt)(\bn-\bp), 
\eeqn
which is equivalent to \eqref{r1}. 
\end{proof}

\commentout{ 
Another useful corollary is this. 
\begin{prop}\cite{FDR} Under the assumptions of Proposition \ref{prop2}, the diffraction pattern of the masked object projection $\widetilde f_\bt$ is non-vanishing  almost surely. 
\label{prop:no-zero}
\end{prop}
}
  
By Corollary \ref{cor1}, for  some $\mbm_\bt \in \IZ^2, \theta_\bt\in \IR$, we have
\beq
\label{1.70}
g_\bt(\bn)\nu(\bn)&=&e^{\im \theta_\bt} f_\bt(\bn + \mbm_\bt  )\mu(\bn + \mbm_\bt  )\\
\nn\mb{\rm or\quad Twin}(g_\bt\nu)(\bn)&=& e^{-\im \theta_\bt} f_\bt(\bn+\mbm_\bt  ) \mu(\bn +\mbm_\bt  ).
\eeq
Since $\mbox{\rm Twin}(g_\bt)(\bn)=\bar g_\bt(-\bn)$,  we rewrite \eqref{1.70} as
\beq
\label{1.7}
g_\bt(\bn)\nu(\bn)&=&\lt\{\begin{matrix}
e^{\im \theta_\bt} f_\bt(\bn + \mbm_\bt  )\mu(\bn + \mbm_\bt ) \\
 e^{\im \theta_\bt} \bar f_\bt(-\bn +\mbm_\bt  )\bar\mu(-\bn +\mbm_\bt  )
 \end{matrix}\rt.
\eeq
\commentout{
which translates into 
\beq
\widehat g_{x(\alpha,\beta)}(j,k)&=& \widehat f(-\alpha j-\beta k, j,k),\quad j,k\in \IZ_{2n-1}.
\eeq
Likewise, we have
\beq
\widehat g_{y(\alpha,\beta)}(i,k)&=& \widehat f(i, -\alpha i-\beta k, k), \quad i,k\in \IZ_{2n-1}\\
\widehat g_{z(\alpha,\beta)}(i,j)&=& \widehat f(i,j,  -\alpha i-\beta j), \quad i,j\in \IZ_{2n-1}. 
\eeq
for  some $\mbm_\bt  \in \IZ,\theta_\bt\in \IR$. 
}

If $\mu$ is completely known, i.e. $\nu=\mu$, then \eqref{1.7} becomes 
\beq
\label{1.8}
g_\bt(\bn)\mu(\bn)&=&\lt\{\begin{matrix}
e^{\im \theta_\bt} f_\bt(\bn + \mbm_\bt  )\mu(\bn + \mbm_\bt ) \\
 e^{\im \theta_\bt} \bar f_\bt(-\bn +\mbm_\bt  )\bar\mu(-\bn +\mbm_\bt  ).
 \end{matrix}\rt.
\eeq

Our goal is to  prove that with a sufficiently large $\cT$,  \eqref{1.7} yields  $g=f$ and $\mu=\nu$, up to a constant phase factor,  almost surely, i.e. $\mbm_\bt=0$ and $\theta_\bt=\mbox{const.}$ for all $\bt$ and eventually design an efficient algorithm to reconstruct $f$.

\section{Uniqueness theorems}\label{sec:unique}

Our first main result is that with the help of a random mask,  tomographic phase retrieval reduces to computed tomography modulo the ambiguity that the object projection  is independent of the direction used in the measurement scheme. 

\commentout{
Let $\mb{\rm Box}[E]$ stand for the box hull,  the smallest
rectangle,  with edges parallel to $\be_1=(1,0)$ or $\be_2=(0,1)$, containing $E$.
 In our present setting an object projection  $f_\bt$ is said to have a tight support if 
\beq
\label{tight}
\mb{\rm Box}[\supp(f_\bt)]=\IZ_n^2
\eeq 
otherwise $f_\bt$ is said  to have  a loose support. Clearly, $f_\bt$ has a tight support if and only if $\mb{\rm Twin}(f_\bt)$ does  since  $\mb{\rm Box}[\supp(f_\bt)]=\mb{\rm Box}[\supp(\mb{\rm Twin}(f_\bt))]+\mbm$ for some $\mbm$.
}
\begin{thm}[Reduction to CT modulo an ambiguity]  \label{tom} Consider a random phase mask $\mu(\bn)=\exp[\im \phi(\bn)]$ with independent, continuous random variables $\phi(\bn)\in \IR$. 
 Suppose that $f_\bt$ is a non-line object  for all $\bt\in \cT$.
If $g$ is supported in $\IZ_n^3$ and produces  the same diffraction patterns as $f$ for all $\bt\in \cT$,  then 
with probability one either
\beq
g_\bt=e^{\im\theta_0} f_\bt,\quad \forall\bt\in \cT
\label{2.40}
\eeq  
or 
\beq\label{2.35}
{g}_{\mathbf{t}}&=&{g}_{\mathbf{t}'}, \quad \forall\bt,\bt'\in \cT,
\eeq
(including the special case ${f}_{\bt}={f}_{\bt'}, \forall\bt,\bt'\in \cT$). 
\label{lem1}
\end{thm}

\commentout{
\begin{rmk}
The $n+1$ directions
$(1,s_i, 0), i=1,\ldots, n$ and $(0,1,0)$ correspond to rotations about the $z$-axis; the $n+1$ directions
$(s_i,0,1), i=1,\ldots, n$ and $(1,0,0)$ correspond to rotations about the $y$-axis; 
the $n+1$ directions
$(0, 1, s_i), i=1,\ldots, n$ and $(0,0,1)$ correspond to rotations about the $x$-axis.  

\end{rmk}
}

\commentout{
\begin{thm} Suppose that $f_\bt$ has a tight support for every $\bt\in \cT$. Let $S$ be a set of slopes with $n$ distinct values, $s_i\in (-1,1), i=1,\ldots, n.$ If $g$ produces  the same $3n$ diffraction patterns in the directions
$(1,s_i, 0)$ (\hbox{or} $(1,0,s_i)$),  $(0,1,s_i)$ (\hbox{or}  $(s_i, 1,0)$) and $(s_i, 0,1)$ (\hbox{or} $(0,s_i,1)$) for $i=1,\ldots, n$,  then either $g=e^{\im\theta_0}\overline{f}$,  in which case both $f$ and $g$ are a $\delta$-function located at the origin, or $g=e^{\im \theta_0} f$.  
Here $\theta_0$ is an arbitrary constant in $\IR$.  
\label{tom}
\end{thm}

\begin{rmk}
The $n+1$ directions
$(1,s_i, 0), i=1,\ldots, n$ and $(0,1,0)$ correspond to rotations about the $z$-axis; the $n+1$ directions
$(s_i,0,1), i=1,\ldots, n$ and $(1,0,0)$ correspond to rotations about the $y$-axis; 
the $n+1$ directions
$(0, 1, s_i), i=1,\ldots, n$ and $(0,0,1)$ correspond to rotations about the $x$-axis.

\end{rmk}
}
\begin{rmk}
 If $f$ is a non-planar object than it follows that $f_\bt$ is a non-line object  for all $\bt$.  

\end{rmk}

\begin{rmk}\label{rmk4.3}
With a plain (instead of random) mask, the twin-object ambiguity $g(\bn)=e^{\im\theta_0}\overline{f(-\bn)}$ can not be eliminated. 

\end{rmk}

\commentout{
\begin{rmk}
One can argue that $n+1$ is very close to the minimum number of diffraction patterns necessary for tomographic phase retrieval as follows. 

Suppose each diffraction pattern is sufficient for unique determination of the corresponding projection of the object (which is false unless  provided with additional prior information, see \cite{unique}), which accounts for $n\times n$ degrees of freedom. Since there are $n^3$ degrees of freedom in $f$, one needs at least $n$ distinct diffraction patterns for unique determination of $f$. 

\end{rmk}

\begin{rmk}

In contrast, for the continuum setting, 
a compactly supported function is uniquely determined by  the Fourier transform (magnitude \& phase) in any infinite set of projections (\cite{Helgason}, Proposition 7.8) while
for any finite set of projections, counterexamples to unique determination can be constructed (\cite{Helgason}, Proposition 7.9).

As a consequence, uniqueness with Fourier magnitude data in the continuum setting  would certainly require an infinite number of projections. It is not currently known, however, that which infinite sets of projections guarantee uniqueness with Fourier magnitude measurements. 

\end{rmk}

\begin{rmk}
Spherically symmetric objects are automatically precluded from our discrete object model.  To create this ambiguity, one has to go to the {\bf continuum} setting and consider the model of 3D objects with finite support \cite{continuum}. 

Note also that due to zero padding, no object can give rise to a $\delta$-function on any Fourier plane. 
\end{rmk}
}

\begin{proof}

\commentout{
Suppose that there is an element $\bt_0\in \cT$ such that 
\beq\label{1.10}
g_{\bt_0}(\bn)&=&
e^{\im \theta_{\bt_0}} f_{\bt_0}(\bn + \mbm_{\bt_0}  )\mu(\bn + \mbm_{\bt_0} )/\mu(\bn),
\eeq
one of the two alternatives stipulated in \eqref{1.8}. Note that \eqref{1.10} along with the 3D support constraint \eqref{1.1} imply that the support of the right hand side of \eqref{1.10} is a subset of $\IZ^2_n$. 

Let $\widetilde f$ be the 3D object defined as
$f$  multiplied by 
$\mu(\cdot)/\mu(\cdot-\mbm_{\bt_0})
$ along the projection lines in the direction $\bt_0$. Hence 
\beq
\widetilde f_{\bt_0}(\bn+\mbm_{\bt_0})= f_{\bt_0}(\bn + \mbm_{\bt_0}  )\mu(\bn + \mbm_{\bt_0} )/\mu(\bn),
\eeq
and $g_{\bt_0}(\bn)=e^{\im \theta_{\bt_0}} \widetilde f_{\bt_0}(\bn+\mbm_{\bt_0})$.

For  other directions $\bt\in \cT$, we can rewrite \eqref{1.8} as 
\beq
\label{1.10}
g_\bt(\bn)&=&\lt\{\begin{matrix}
e^{\im \theta_\bt} \widetilde f_\bt(\bn+\mbm_{\bt} )\\
 e^{\im \theta_\bt} \bar f_\bt(-\bn +\mbm_\bt)\bar\mu(-\bn+\mbm_\bt )/\mu(\bn).
 \end{matrix}\rt.
\eeq
}

Suppose that,  for some $\bt_0\in \cT$, the first alternative in \eqref{1.8} holds true,  i.e.
\beq\label{1.10}
g_{\bt_0}(\bn)&=&
e^{\im \theta_{\bt_0}} f_{\bt_0}(\bn+\mbm_{\bt_0})\lamb_{\bt_0}(\bn+\mbm_{\bt_0})
\eeq
with 
\[\lamb_{\bt_0}(\bn)=\mu(\bn)/\mu(\bn-\mbm_{t_0}),
\]
implying
\[
\widehat{g_{\bt_0}}=e^{\im \theta_{\bt_0}}e^{\im 2\pi \mbm_{\bt_0} \cdot\bk/p} \widehat f_{\bt_0}\star \widehat\lamb_{\bt_0}(\bk).
\]

We now prove that the second alternative in \eqref{1.8} can not hold. Otherwise, suppose that for some $\bt\in \cT$,
\beq
\label{1.27}
g_\bt(\bn)&= & e^{\im \theta_\bt} \overline{ f_\bt(-\bn +\mbm_\bt)\nu_\bt(-\bn+\mbm_\bt)}
\eeq
with
\[
\nu_\bt(\bn)={ \mu(\bn)}/\overline{\mu(-\bn+\mbm_{t})}. 
\]
implying  
\[
\widehat{g_\bt}(\bk)
=\overline{\widehat f_\bt\star \widehat \nu_\bt(\bk )} e^{-\im 2\pi \mbm_\bt \cdot\bk/p}
\]
where $\star$ denotes the discrete convolution over the periodic grid $\IZ_p^2$.

Let $P_\bt$ denote the origin-containing (continuous) plane orthogonal to $\bt$ in the Fourier space. 
By Fourier slice theorem, for all $\bk\in P_\bt \cap P_{\bt_0}$, $\widehat g_{\bt_0}(\bk)=\widehat g_\bt(\bk)$ and hence
\[
e^{\im \theta_{\bt_0}}e^{\im 2\pi \mbm_{\bt_0} \cdot\bk/p} \widehat f_{\bt_0}\star \widehat\lamb_{\bt_0}(\bk)= e^{\im \theta_\bt} e^{-\im 2\pi \mbm_\bt \cdot\bk/p}\overline{\widehat f_\bt\star \widehat \nu_\bt }(\bk),\quad\forall\bk\in P_\bt \cap P_{\bt_0}
\]
implying 
\beq\label{1.30}
&&e^{\im \theta_{\bt_0}}e^{\im 2\pi \mbm_{\bt_0} \cdot\bk/p}\sum_{\bn\in \IZ^2_n} e^{\im \phi(\bn)}
e^{-\im \phi(\bn-\mbm_{\bt_0})}f_{\bt_0}(\bn)e^{-\im2\pi \bn\cdot\bk/p}\\
&= &e^{\im \theta_\bt} e^{-\im 2\pi \mbm_\bt \cdot\bk/p} \sum_{\bn\in \IZ_n^2} e^{-\im \phi(\bn)}e^{-\im\phi(-\bn+\mbm_\bt)}\bar f_\bt(\bn)e^{\im 2\pi \bn\cdot\bk/p},\quad\forall\bk\in P_\bt \cap P_{\bt_0}. \nn
\eeq

We now show that eq. \eqref{1.30} can not hold for any $\mbm_{\bt_0}, \mbm_\bt$. Consider any $\bn$  that $f_{\bt_0}(\bn)\neq 0$. Due to the statistical independence of $\phi(\cdot)$ and the sign of the phases in 
\[
e^{\im \phi(\bn)}
e^{-\im \phi(\bn-\mbm_{\bt_0})}f_{\bt_0}(\bn)e^{-\im2\pi \bn\cdot\bk/p},\quad 
e^{-\im \phi(\bn)}e^{-\im\phi(-\bn+\mbm_\bt)}\bar f_{\bt}(\bn)e^{\im 2\pi \bn\cdot\bk/p}
\]
appearing in the summation on either side of \eqref{1.30},
the phase factor $e^{\im \phi(\bn)}$ on the left can not be balanced without setting $\mbm_{\bt_0}=0$. This then implies
\[e^{\im \theta_{\bt_0}}f_{\bt_0}(\bn)e^{-\im2\pi \bn\cdot\bk/p}=e^{\im \theta_\bt} e^{-\im 2\pi \mbm_\bt \cdot\bk/p}
e^{-\im \phi(\bn)}e^{-\im\phi(-\bn+\mbm_\bt)}\bar f_\bt(\bn)e^{\im 2\pi \bn\cdot\bk/p}
\]
which cannot hold since the left hand size is deterministic while the right hand side is random. Consequently,  \eqref{1.27} is false almost surely, which leaves the first of \eqref{1.8} the only viable alternative.

If, however, $f_{\bt_0}(\bn)=0$ for all $\bn$, then the same argument implies that $f_\bt(\bn)=0$ for all $\bn$. 
Hence the first alternative of \eqref{1.8} still follows,
i.e.
\beq
\label{1.13}
g_\bt(\bn)&=&
e^{\im \theta_\bt} f_\bt(\bn+\mbm_\bt )\lamb_\bt(\bn+\mbm_\bt),\quad \forall \bt \in \cT, 
\eeq
{for some} $\mbm_\bt$.

Consider two arbitrary,  distinct directions  $\bt=\bt_1,\bt_2\in \cT$. 
By the Fourier slice theorem,   
\beq
e^{\im\theta_{\bt_1}}e^{\im 2\pi \mbm_{\bt_1} \cdot\bk}\widehat f_{\bt_1} \star\widehat\lamb_{\bt_1}(\bk)&=&
e^{\im\theta_{\bt_2}}e^{\im 2\pi \mbm_{\bt_2} \cdot\bk}\widehat f_{\bt_2} \star\widehat\lamb_{\bt_2}(\bk),\quad\forall \bk\in P_{\bt_1} \cap P_{\bt_2}\label{1.61}
\eeq
implying 
\beq\label{1.30'}
&&e^{\im \theta_{\bt_1}}e^{\im 2\pi \mbm_{\bt_1} \cdot\bk/p}\sum_{\bn\in \IZ^2_n} e^{\im \phi(\bn)}
e^{-\im \phi(\bn-\mbm_{\bt_1})}f_{\bt_1}(\bn)e^{-\im2\pi \bn\cdot\bk/p}\\
&=&\nn
e^{\im \theta_{\bt_2}}e^{\im 2\pi \mbm_{\bt_2} \cdot\bk/p}\sum_{\bn\in \IZ^2_n} e^{\im \phi(\bn)}
e^{-\im \phi(\bn-\mbm_{\bt_2})}f_{\bt_2}(\bn)e^{-\im2\pi \bn\cdot\bk/p},\quad\forall \bk\in P_{\bt_1} \cap P_{\bt_2}. 
\eeq
Due to the statistical independence of $\phi(\cdot)$, the randomness on the both sides of \eqref{1.30'} can not balance out unless
\beq
\mbm_{\bt_1}&=&\mbm_{\bt_2} (=\mbm_0) \label{2.29}\\
\lamb_{\bt_1}&=&\lamb_{\bt_2}.\label{2.30}
\eeq
Eq. \eqref{2.30} means independence of $\lambda_\bt$ from $\bt\in \cT$ and justifies the simplified notation
\beq
\label{2.33}
\lamb_\bt(\bn)=\lamb_0(\bn):=\mu(\bn)/\mu(\bn-\mbm_0)\quad\hbox{for some $\mbm_0\in \IZ^2$ and  all $\bt\in \cT$}. 
\eeq
With this, \eqref{1.30'} reduces to
\beq
e^{\im \theta_{\bt_1}}\sum_{\bn\in \IZ^2_n} e^{\im \phi(\bn)}
e^{-\im \phi(\bn-\mbm_{0})}f_{\bt_1}(\bn)e^{-\im2\pi \bn\cdot\bk/p}
&=&e^{\im \theta_{\bt_2}}\sum_{\bn\in \IZ^2_n} e^{\im \phi(\bn)}
e^{-\im \phi(\bn-\mbm_{0})}f_{\bt_2}(\bn)e^{-\im2\pi \bn\cdot\bk/p},\label{2.32}
\eeq
for all $ \bk\in P_{\bt_1} \cap P_{\bt_2}. $

The function $\lamb_0$ defined in \eqref{2.33} is either  $1$ (if $\mbm_0=0$) or random (if $\mbm_0\neq 0$). 
If $\mbm_0=0$, then, by \eqref{1.13}, 
$g_\bt=e^{\im\theta_\bt} f_\bt$ for all $\bt\in \cT$. By Fourier slice Theorem, 
\beqn
\widehat f_\bt(\bk)&=&\widehat f_{\bt_0}(\bk),\quad\forall\bk\in P_\bt \cap P_{\bt_0}\\
\widehat g_\bt(\bk)&=&\widehat g_{\bt_0}(\bk), \quad\forall\bk\in P_\bt \cap P_{\bt_0}
\eeqn
and hence $\theta_\bt=\theta_{0}$ for some $\theta_0\in \IR$ and all $\bt\in \cT$.
In other words, $
g_\bt=e^{\im\theta_0} f_\bt, \forall\bt\in \cT.$

If $\mbm_0\neq 0$, then \eqref{2.32} and  the statistical independence  of $\phi(\cdot)$ imply that
\beq\label{2.34}
e^{\im\theta_{\bt_1}} \widehat{f}_{\mathbf{t_1}}&=&e^{\im\theta_{\bt_2}}\widehat{f}_{\mathbf{t_2}},\quad\forall \bt_1,\bt_2\in \cT. 
\eeq
Hence by \eqref{1.13} and \eqref{2.29}
\beq\nn
{g}_{\mathbf{t_1}}&=&{g}_{\mathbf{t_2}}, \quad \forall\bt_1,\bt_2\in \cT,
\eeq
almost surely. In other words, $ g_{\mathbf{t}}$ is independent of $\mathbf{t}\in \cT$ with probability one.

Let us turn to  the remaining undesirable alternative:
\beq
\label{1.14}
g_\bt(\bn)\mu(\bn)&=&
 e^{\im \theta_\bt} \overline{ f_\bt(-\bn +\mbm_\bt) \mu(-\bn+\mbm_\bt)},\quad \forall \bt \in \cT
\eeq
or, equivalently, \eqref{1.27}.

For two distinct projections $\bt=\bt_1,\bt_2\in \cT$, 
\eqref{1.27} implies 
\beq\label{1.29}
e^{-\im\theta_{\bt_1}}{{\widehat f_{\bt_1}\star \widehat \nu_{\bt_1}}(\bk)}
=e^{-\im\theta_{\bt_2}}{\widehat f_{\bt_2}\star \widehat \nu_{\bt_2}(\bk)},\quad \forall \bk\in\cP_{\bt_1}\cap\cP_{\bt_2}. 
\eeq
\commentout{By the Fourier slice theorem, $
\widehat g_{\bt_1}=\widehat g_{\bt_2}$ and $\widehat f_{\bt_1}=\widehat f_{\bt_2}$  along the common line $\bt_1\times \bt_2$ and hence $\theta_{\bt_1}=\theta_{\bt_2}$, i.e. $\theta_\bt =\theta_0$ for some $\theta_0$  independent of $\bt\in \cT$. In other words,  \eqref{1.14} becomes 
\beq
\label{1.28}
g_\bt(\bn)\mu(\bn)&=&
 e^{\im \theta_0} \overline{ f_\bt(-\bn ) \mu(-\bn)},\quad \forall \bt \in \cT.
\eeq
\eqref{1.27} implies 
\beq\label{1.29}
e^{\im\theta_{\bt_1}}e^{\im 2\pi \mbm_{\bt_1}\cdot\bk}\overline{{\widehat f_{\bt_1}\star \widehat \mu}(\bk)}
=e^{\im\theta_{\bt_2}}e^{\im 2\pi \mbm_{\bt_2}\cdot\bk}\overline{\widehat f_{\bt_2}\star \widehat \mu(\bk)},\quad \forall \bk\in \bt_1\times \bt_2. 
\eeq
On the other hand, 
by \eqref{1.60}-\eqref{1.61}
 \beq\label{1.29'}
e^{\im\theta_{\bt_1}}e^{\im 2\pi \mbm_{\bt_1}\cdot\bk}=e^{\im\theta_{\bt_2}}e^{\im 2\pi \mbm_{\bt_2}\cdot\bk},\quad \forall \bk\in \bt_1\times \bt_2
\eeq
implying
\beq\label{1.37}
(\mbm_{\bt_1}-\mbm_{\bt_2})\perp (\bt_1\times\bt_2)\quad\hbox{and}\quad \theta_{\bt_1}=\theta_{\bt_2},\quad\forall\bt\in \cT. 
\eeq
}
which, at $\bk=0$, means
\beqn
e^{-\im\theta_{\bt_1}} \sum_{j,k\in \IZ_{n}}\widehat{f}_{\mathbf{t_1}}(j,k)\widehat{\nu_{\bt_1}}(-j,-k)&=&e^{-\im\theta_{\bt_1}}\sum_{j,k\in \IZ_{n}}\widehat{f}_{\mathbf{t_2}}(j,k)\widehat{\nu_{\bt_2}}(-j,-k).
\eeqn
The rest of the argument follows exactly the same pattern as that following \eqref{2.32}.

\end{proof}

In view of Theorem \ref{lem1}, with a randomly coded aperture, the uniqueness problem of phase retrieval is only slightly more difficult than that of computed tomography, with only the additional ambiguity \eqref{2.35} to resolve. 

First  let us digress and consider some generic schemes that guarantee uniqueness for computed tomography. 

 \begin{ex}\label{ex1}
{\rm  Let  $\cT$ consist of the projections represented as \eqref{2.10}:
\beq
\cT&=&\{(\alpha_l,\beta_l, 1): l=1,\dots, m\},\label{50}
\eeq
for some $m\in \IZ$ and suppose \eqref{2.40} holds, i.e.
\beqn
g_\bt=e^{\im\theta_0} f_\bt,\quad\forall\bt\in \cT.
\eeqn
 By the Fourier slice theorem, we have
\beq
\label{1.14'}
\widehat{g}(j,k, -\alpha_l j-\beta_l k) &= &e^{\im \theta_0}\widehat{f}(j,k,-\alpha_l j-\beta_l k),\quad l=1,\dots, m,
\eeq
where both $\widehat g(j,k,\cdot)$ and $\widehat f(j,k,\cdot)$ are $p$-periodic signals bandlimited to $\pi (n-1)/p$ (for odd integer $n$, cf. \eqref{1.2}). 
In order to conclude that $\widehat g=e^{\im \theta_0}\widehat f$, it suffices to have  
\beq
\label{strongCT}
|\{\alpha_l j+\beta_l k \,\,\hbox{(mod $p$)}: l=1,\dots, m\}|  \ge n,\quad \forall (j,k)\neq (0,0), 
\eeq
which is also a necessary condition for the validity of $\widehat g=e^{\im \theta_0}\widehat f$,  in general (see \cite{Eldar}).
}
\end{ex}

{Slightly modifying the observation  in Example \ref{ex1}, we can state the  following uniqueness theorem for 3D discrete computed tomography. 
\begin{thm}[\rm Uniqueness of CT]\label{thm:CT}
 Let  $\cT$ be {\bf any} one of the following three sets of projections:
\beqn
(x)&&\{(1, \alpha_l,\beta_l): l=1,\dots,m \}\\
(y)&&\{(\alpha_l,1, \beta_l): l=1,\dots,m \}\\
(z)&&\{(\alpha_l,\beta_l, 1): l=1,\dots,m \}.
\eeqn
Then
$g=e^{\im \theta_0} f$,  whenever $g_\bt =e^{\im \theta_0} f_\bt$ for all $\bt\in \cT$ and some constant $\theta_0\in \IR$,  if and only if 
the condition \eqref{strongCT} holds true.

\end{thm}
\begin{rmk}\label{rmk:CT}
The condition \eqref{strongCT} can be achieved with overwhelming probability by randomly and independently selecting $n$ pairs of $(\alpha_l,\beta_l)$ (i.e. $m=n$) with the uniform distribution over the square $|\alpha_l|, |\beta_l| <1,$\cite{rand}. 


In view of the Fourier slice theorem, the redundancy in 3D discrete CT due to the overlap of Fourier planes with different normal vectors (i.e. the common lines) can be roughly estimated as follows. Every pair of Fourier planes share a common line of  about $n$ degrees of freedom. There are in general $n(n-1)/2$ pairs from $n$ distinct Fourier planes and hence $n^2(n-1)/2$ degrees of information overlap.   As oversampling the projection planes (cf. \eqref{nyquist1} \& \eqref{nyquist2}) compensates  the information overlap, $n$ generic  projections  contain sufficient information for determining the $n^3$  degrees of freedom in the object.   

\end{rmk}
}

In  X-ray diffractive imaging, a most commonly used scheme is rotated projections about an axis orthogonal to the directions of projection. For example,  \beqn
&&\{(\alpha_l,0,  1): l=1,\ldots, m\},\label{T1}
\eeqn
where $\alpha_l$ are distinct numbers, 
represents a sequence of projections rotated about the $y$-axis. 
More generally, rotated projections forming the same angle $\arctan(\gamma)$ with, say,  the $z$-axis, can be represented as 
\beqn
&&\{(\gamma\cos t_j, \gamma\sin t_j, 1): j=1,\dots,m \}. \label{scheme1}
\eeqn

Next,  we demonstrate that with one additional projection to the scheme such as  in Example \ref{ex1}, one can eliminate the possibility \eqref{2.35} and resolve the uniqueness problem for tomographic phase retrieval. 

\begin{ex} [\rm Resolution of ambiguity \eqref{2.35}] \label{Yu} {\rm  Let  $\cT$ consist of the projections represented as \eqref{2.8}:
\beqn
\cT&=&\{(1, \alpha_l,\beta_l): l=1,\dots, n\}\cup \{(0,\alpha_0,\beta_0)\}\label{T2}
\eeqn
satisfying \eqref{strongCT} and $(\alpha_0,\beta_0)\neq (0,0).$

In terms of  the X-ray transform, \eqref{2.35} means that, for some $c(\cdot,\cdot)$ independent of $\alpha,\beta,$
\beq\label{3.33}
\widehat{g}_{x(\alpha, \beta)}(j,k) &= &c(j,k) 
\eeq
and hence by Fourier Slice Theorem
\beq
\label{3.45}
\widehat g(-\alpha j-\beta k, j, k)&=&c(j,k)
\eeq
 for  $ j,k\in \IZ_{p}$.

Let
\beq
\label{xg}
\widehat{g}(\xi,\eta,\zeta)= \sum_{m}\widehat{g}_{\eta\zeta}(m) e^{-2\pi \im m\xi /{p}}
\eeq
with
\beq
\widehat{g}_{\eta\zeta}(m)&=&\sum_{l} \widehat g_\eta (m,l) e^{-2\pi \im l\zeta/{p}}\label{3.36} 
\eeq
and 
\beq
\label{3.35}
\widehat{g}_{\eta}(m,l)&=&\sum_{k} g(m,k,l) e^{-2\pi \im k\eta/{p}}.
\eeq

By the support constraint $\supp(g)\in \IZ_n^3$, \eqref{xg} becomes the $n\times n$ Vandermonde  system 
\beq
\label{van'}
V \widehat g_{\eta\zeta}=\lt[\begin{matrix} c(\eta, \zeta)\\ 
c(\eta, \zeta)\\ 
\vdots\\
c(\eta,\zeta)
\end{matrix}
\rt].
\eeq
with
\beq
V&=&[V_{ij}],\quad V_{ij}= e^{-2\pi \im \xi_ij/{p}},\quad \xi_i=-\alpha_i\eta-\beta_i\zeta.
\eeq
which is nonsingular if and only if $\{\xi_i: i=1, \dots,n\}$ has $n$ distinct members. 

Since the system \eqref{van'} has a unique solution for $(\eta,\zeta)\neq (0,0)$,  we identify 
$\widehat g_{\eta\zeta}(\cdot)$ to be the discrete $\delta$-function located at $0$ with amplitude $c(\eta,\zeta)$ for $(\eta,\zeta)\neq (0,0)$.

For $\eta,m\neq 0$, $\widehat{g}_{\eta\zeta}(m)=0$ for all $\zeta$ and hence $\widehat g_{\eta}(m,l)=0 $ for all $l$. Likewise for \eqref{3.35}, we select $n$ distinct, nonzero values for $\eta$ to perform inversion of the Vandermonde system and obtain   
\beq
g(m, k, l)=0, \quad m\neq 0. \label{1.41}
\eeq
 In other words, $g$ is supported on the $y-z$ plane. 
Consequently  the projection in the direction of $(0,\alpha_0,\beta_0)$ of $g$ would be a line object,  contradicting to the assumption of non-line projection in Theorem \ref{lem1}. Therefore, \eqref{2.35} is false and \eqref{2.40} holds true almost surely for  the scheme $\cT$ under the assumptions of Theorem \ref{tom}. 
 }
 \end{ex}
 
Slightly  extending the above  analysis, we are ready to state the final result. 
\begin{thm}\label{tom2}
 Let  $\cT$ be {\bf any} one of the following three sets of projections:
\beqn
(x')&&\{(1, \alpha_l,\beta_l): l=1,\dots,n \}\cup \{(0, \alpha_0,\beta_0)\}\\
(y')&&\{(\alpha_l,1, \beta_l): l=1,\dots,n \}\cup \{(\alpha_0,0, \beta_0)\}\\
(z')&&\{(\alpha_l,\beta_l, 1): l=1,\dots,n \}\cup \{(\alpha_0,\beta_0,0)\}
\eeqn
satisfying condition \eqref{strongCT} (with $m=n$) and $(\alpha_0,\beta_0)\neq (0,0)$. 
\commentout{
with the property that 
\beq
\label{100}
&&(\alpha_0,\beta_0)\neq (0,0)\quad \& \\
&&\{\alpha_l \xi+\beta_l\eta: l=1,\dots, n\} \hbox{ has $n$ distinct members for each fixed $(\xi,\eta)\neq (0,0)$}. \label{101}
\eeq
}
Then under the assumptions of Theorem \ref{tom}, we have 
$g=e^{\im \theta_0} f$,  for some constant $\theta_0\in \IR$,  with probability one.

\end{thm}

\section{Conclusion and discussions}\label{sec:conclude}

The key to our approach is Theorem \ref{tom} which essentially reduces 3D discrete tomographic phase retrieval to computed tomography (CT).

{Uniqueness condition for CT (Theorem \ref{thm:CT}) sets a lower bound $n$ on the number of diffraction patterns needed for tomographic phase retrieval since  each diffraction pattern contains no more information than the corresponding projection ($f_\bt$ determines the autocorrelation of $f_\bt$ but not vice versa). Therefore, Theorems  \ref{tom2} is nearly, if not exactly,  sharp in terms of the required number of diffraction patterns}.

\commentout{
To contrast with the continuum setting, 
a compactly supported function is uniquely determined by  the Fourier transform (magnitude \& phase) in any infinite set of projections (\cite{Helgason}, Proposition 7.8) while
for any finite set of projections, counterexamples to unique determination can be constructed (\cite{Helgason}, Proposition 7.9).

As a consequence, uniqueness with Fourier magnitude data in the continuum setting  would certainly require an infinite number of projections. It is not currently known, however, that which infinite sets of projections guarantee uniqueness with Fourier magnitude measurements. 
}

\commentout{  It is worthwhile to compare our results with those from geometric tomography in which the object function is {\em the indicator function of a convex set}. It is shown that  the autocorrelation function uniquely determines, up to translation and reflection, 2D convex bodies \cite{AB09} as well as 3D  convex polyhedra \cite{Bianchi09}.  It is still unknown whether all 3D convex bodies  are determined (most are in the sense of Baire category \cite{covex-proj}). Bianchi \cite{Bianchi05}, however, constructed  convex polytopes in $\IR^d, d\ge 4$, that are not determined by autocorrelation. Moreover, there exist noncongruent nonconvex polygons, even horizontally and vertically convex polyominoes, with the same covariogram, indicating that the convexity assumption also cannot be significantly weakened \cite{Gardner}.}


{On the other hand, Theorem \ref{tom2} (condition \eqref{strongCT} in particular)  defines a fairly general class of measurement schemes. A natural question is, Which one is optimal and in what sense? This will be the subject of our forthcoming study. 
}

{  In realistic measurements, noise is inevitable. And because of the significant amount of oversampling (cf. \eqref{nyquist1}-\eqref{nyquist2}), independent noise in the data necessarily results in an inconsistent  inverse problem, i.e. there is no object whose tomographic data coincide with  the given noisy data. This is characteristic of the ill-posedness of inverse problems in general. Noise stability analysis for tomographic phase retrieval is technically challenging and currently lacking. In practice, however,  noisy reconstruction can often be effectively performed by utilizing prior information and regularization such as Tikhonov regularization \cite{2018}. 

Other useful regularizations include sparsity-promoting priors such as $\ell_1$ and total variation  regularizations. 
In our setting, for a sparse object whose projection $f_\bt$ is supported on a  much smaller set than $\IZ_p^2$, the diffraction pattern can be measured at a comparably small (up to a poly-logorithmic factor of $n$), randomly selected subset of $\IZ_{2p-1}^2$ from which the autocorrelation of $f_\bt$ can be recovered by $\ell_1$-minimization method  with the random partial Fourier matrix as the sampling matrix in \eqref{auto} (see \cite{Candes2,Trop}). The total-variation regularization can be used for gradient-sparse objects  \cite{TV}. Similar approaches have been implemented in 3D digital holography \cite{2016}, \cite{2020}.  }



\section*{Acknowledgments}
I thank Qi Yu for helpful discussions about Example \ref{Yu}.  The research is supported by the Simons Foundation grant FDN 2019-24 and the NSF grant CCF-1934568.

\commentout{ 
\begin{appendix}

\section{Sparse sampling with ESPRIT}\label{app:sparse}

For simplicity of notation, we present sparse reconstruction of diffraction patterns by the method of {\em Estimation of Signal Parameters via Rotational Invariance Techniques} (ESPRIT) in the one-dimensional setting. 

Consider the 1D diffraction pattern written as
\beq
|\widehat R(w)|^2= \sum_{k\in \IZ_{2p-1}} R(k)  e^{-\im 2\pi k w}, \quad w\in  \Big[-\half,\half \Big],
   \label{1D}
   \eeq
   where
   \[
 R(k)=  \lt\{\sum_{k'\in \IZ_p} f_\bt(k'+k)\overline{f_\bt(k')}\rt\}
\]
is the autocorrelation function of $f_\bt$. 

Suppose $R$ has a sparse support in the sense that the cardinality  $s$ of $\cS:=\supp(R)$ is much smaller than $2p-1$. We want to demonstrate that the full diffraction pattern can be {\em stably}  reconstructed  from $2s-1$ samples. 
\end{appendix}
}

\end{document}